\documentclass[lettersize,journal]{IEEEtran}
\IEEEoverridecommandlockouts

\usepackage{amssymb}
\usepackage[cmex10]{amsmath}
\usepackage{stfloats}
\usepackage{graphicx}
\usepackage{subfigure}
\usepackage{tabularx}
\usepackage{epsfig,epsf,color,balance,cite}
\usepackage{verbatim}
\usepackage{url}
\usepackage{bm}
\usepackage{booktabs}

\usepackage{amsmath}
\allowdisplaybreaks[2]
\usepackage{ntheorem}
\newtheorem{theorem}{Theorem}

\newtheorem*{proof}{Proof}

\usepackage{algorithm}
\usepackage{algorithmic}

\usepackage{color}
\definecolor{myc1}{rgb}{0,0,0}

\begin{document}

\title{Probabilistic Semantic Communication over Wireless Networks with Rate Splitting}

\author{Zhouxiang Zhao,~\IEEEmembership{Graduate Student Member,~IEEE,}
        Zhaohui Yang,~\IEEEmembership{Member,~IEEE,}
        Ye Hu,~\IEEEmembership{Member,~IEEE,}
        Qianqian Yang,~\IEEEmembership{Member,~IEEE,}
        Wei Xu,~\IEEEmembership{Senior Member,~IEEE,}
        and Zhaoyang Zhang,~\IEEEmembership{Senior Member,~IEEE}
\thanks{Z. Zhao, Z. Yang, Q. Yang and Z. Zhang are with College of Information Science and Electronic Engineering, Zhejiang University, and also with Zhejiang Provincial Key Laboratory of Info. Proc., Commun. \& Netw. (IPCAN), Hangzhou, China (e-mails: \{zhouxiangzhao, yang\_zhaohui, qianqianyang20, ning\_ming\}@zju.edu.cn).}
\thanks{Y. Hu is with Department of Industrial and Systems Engineering, University of Miami, Coral Gables, FL, 33146, USA (e-mail: yehu@miami.edu).}
\thanks{W. Xu is with National Mobile Communications Research Laboratory, Southeast University, Nanjing, China (e-mail: wxu@seu.edu.cn).}
}



\maketitle

\begin{abstract}
In this paper, the problem of joint transmission and computation resource allocation for probabilistic semantic communication (PSC) system with rate splitting multiple access (RSMA) is investigated. In the considered model, the base station (BS) needs to transmit a large amount of data to multiple users with RSMA. Due to limited communication resources, the BS is required to utilize semantic communication techniques to compress the large-sized data. The semantic communication is enabled by shared probability graphs between the BS and the users. The probability graph can be used to further compress the transmission data at the BS, while the received compressed semantic information can be recovered through using the same shared probability graph at each user side. The semantic information compression progress consumes additional computation power at the BS, which inevitably decreases the transmission power due to limited total power budget. Considering both the effect of semantic compression ratio and computation power, the semantic rate expression for RSMA is first obtained. Then, based on the obtained rate expression, an optimization problem is formulated with the aim of maximizing the sum of semantic rates of all users under total power, semantic compression ratio, and rate allocation constraints. To tackle this problem, an iterative algorithm is proposed, where the rate allocation and transmit beamforming design subproblem is solved using a successive convex approximation method, and the semantic compression ratio subproblem is addressed using a greedy algorithm. Numerical results validate the effectiveness of the proposed scheme.
\end{abstract}

\IEEEpeerreviewmaketitle

\section{Introduction}
\IEEEPARstart{D}{ue} to the rapidly increasing demand for emerging applications including metaverse, vehicle-to-everything (V2X), etc., existing wireless communication systems are facing the challenge of low latency and high energy efficiency. To address this challenge, researchers have proposed a new communication paradigm called semantic communication \cite{10024766}.
In the traditional communication system, the receiver wants to recover the same bits sent by the transmitter through decoding the received information.
Different from such bit-intended communication systems, semantic communication shares a same or similar knowledge base between the transmitter and the receiver, to allow resource-efficient small-sized semantic information transmission with the receiver recovering the semantic information with the help of knowledge base \cite{9679803,chaccour2022less,qin2021semantic,9955525,10233741}. 
Using such distinctive framework, semantic communication accrues notable advantages: high communication efficiency that is achieved by small-sized information transmission \cite{10032275}; enhanced security by requiring the prior information of the knowledge base \cite{yang2023secure}; universal communication with the capability of human to machine communication, as the semantic-level information is transmitted in the same understanding space \cite{9864327}. 
However, to fully explore the potential of semantic communication, it is of importance to investigate the multiple access schemes \cite{clerckx2024multiple}.
As for the semantic rate expression, there are two main directions of investigating semantic communication with multiple access, i.e., simulation-based semantic rate and analysis-based semantic rate. 

In simulation-based semantic rate, the semantic rate is obtained through data analytic method using the simulation results \cite{9763856,10375768,10073623,10355766}. 
In \cite{9763856}, the semantic spectral efficiency (S-SE) is defined for text data, and the joint resource allocation of semantic symbol length and channel assignment is optimized for the S-SE maximization. 
Using the model in \cite{9763856} and considering both conventional bit and semantic information transmission, a non-orthogonal multiple access (NOMA) assisted semantic communication system is studied through optimizing power allocation \cite{9953095,10347482}. 
Taking the quality of experience as performance indicator, the authors in \cite{10303275} optimize power allocation and common message allocation for rate splitting multiple access (RSMA) assisted semantic communication. 
Moreover, in \cite{10287956}, the authors optimized the sum semantic rate through optimizing the semantic information ratio, time duration, rate splitting, and transmit beamforming for ultra-reliability and low-latency communication scenarios.
Besides, a deep learning-based multiple access scheme is proposed with considering both joint source-channel coding and signal modulation schemes in \cite{10288558}.
Further in \cite{10122232}, the deep reinforcement learning based algorithm is proposed to solve the dynamic resource allocation problem for the semantic communication system. 
Through using the semantic relay, the semantic rate maximization problem is formulated through optimizing power allocation and channel assignment \cite{hu2023multiuser}. 
However, the performance of the above schemes is always limited by accuracy of the simulation-based semantic rate expression. 

In analysis-based semantic rate, the semantic rate is derived through modifying the Shannon capacity \cite{clerckx2024multiple} with probability theory. 
Using probability graph, the semantic rate expression is derived in \cite{zhao2023joint,10333452} with considering the trade-off between communication and computation.
Further, the joint communication and computation resource allocation framework is proposed for semantic assisted mobile edge computing systems in \cite{10287133}. 
However, the above works are limited to the orthogonal multiple access with semantic communication \cite{zhao2023joint,10333452,10287133}, even though RSMA can provide extra spectral efficiency \cite{9831440,10038476,mao2018rate}. 

Inspired by the above analysis, the main novelty of this paper is a probability graph-enabled multi-user probabilistic semantic communication (PSC) network with RSMA. The key contributions are summarized as follows:
\begin{itemize}
    \item We propose a PSC framework to allow high-rate large data transmission in communication resource constrained networks. In particular, the base station (BS) uses semantic communication techniques to compress and transmit the large-sized data under limited communication resource budget. To do so, the BS shares a probability graph with its served users, such that each served user can recover the original information from compressed ones. Meanwhile, the BS transmits the compressed information to multiple users with the RSMA scheme to fully exploit the performance gain of multiple antennas.
    \item To achieve the high overall semantic rate, a policy needs to be designed for the management of beamforming, rate splitting, and semantic compression with the limited communication and computation resources in the system. We formulate this joint transmission and computation resource allocation problem in an optimization setting, whose goal is to maximize the sum of semantic rate across all users while meeting the total power budget, semantic compression ratio, and rate allocation requirements.
    \item To address this problem, an iterative algorithm is proposed in which the rate allocation and transmit beamforming design subproblem as well as the semantic compression ratio subproblem are alternately optimized. The rate allocation and transmit beamforming design subproblem is solved using a successive convex approximation (SCA) method. The semantic compression ratio subproblem is addressed using a greedy algorithm, where the closed-form solution for the optimal semantic compression ratio is obtained at each iteration. Numerical results validate the effectiveness of the proposed algorithm.
\end{itemize}

The rest of this paper is organized as follows. The system model and problem formulation are described in Section \uppercase\expandafter{\romannumeral2}. The algorithm design is presented in Section \uppercase\expandafter{\romannumeral3}. Simulation results are analyzed in Section \uppercase\expandafter{\romannumeral4}. Conclusions are finally drawn in Section \uppercase\expandafter{\romannumeral5}.

\section{System Model and Problem Formulation}

\begin{figure}[t]
\centering
\includegraphics[width=0.9\linewidth]{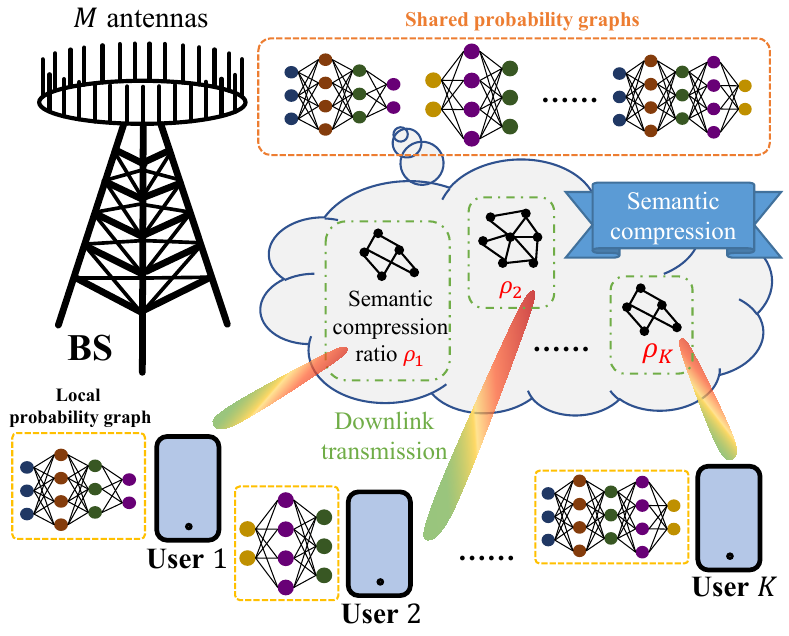}
\caption{Illustration of the considered PSC network.}
\label{fg:sm}
\end{figure}

Consider a downlink wireless PSC network with one multiple-antenna BS and $K$ single-antenna users, as shown in Fig.~\ref{fg:sm}. The BS is equipped with $M$ antennas, and the set of users is denoted by $\mathcal{K}$. The BS needs to transmit a large-sized data set $\mathcal{D}_k$ to each user $k$. Considering the wireless resource limitations, the BS extracts the small-sized semantic information $\mathcal{C}_k$ from the original data $\mathcal{D}_k$ to reduce the load of communication tasks. In other words, the BS in the PSC network determines how the large-sized data are compressed and transmitted to the distributed users, based on the network resource limitations, and user needs. Next, the probability graph enabled data compression framework and RSMA enabled data transmission solution are introduced, followed by the analysis on their transmission performance and energy consumption. Then, the problem formulation is given.

\subsection{Semantic Communication Model}

\begin{figure}[t]
\centering
\includegraphics[width=0.8\linewidth]{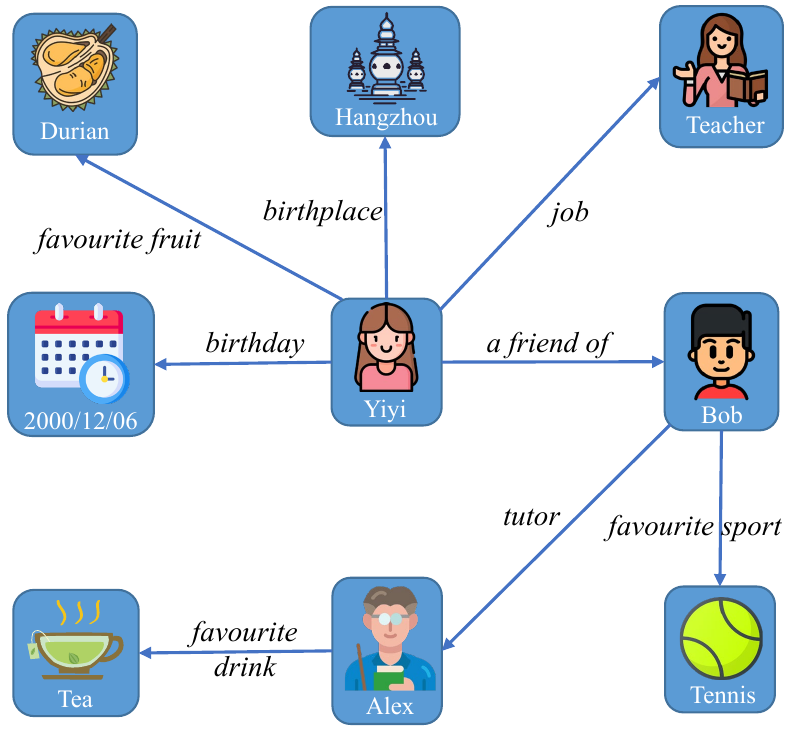}
\caption{An example of a conventional knowledge graph.}
\label{fg:kg}
\end{figure}

\begin{figure}[t]
\centering
\includegraphics[width=\linewidth]{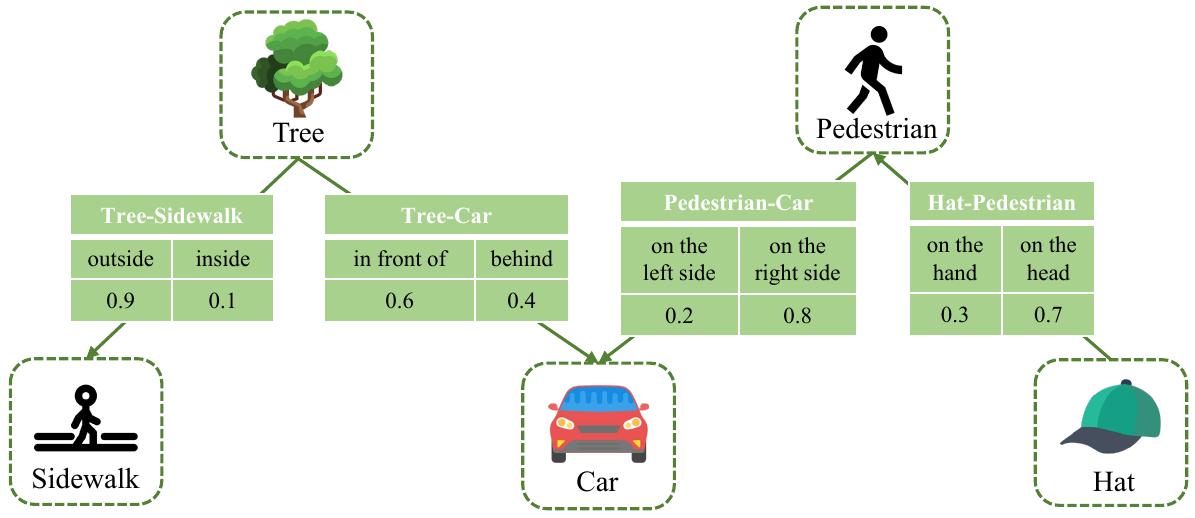}
\caption{Illustration of the probability graph considered in the PSC system.}
\label{fg:pg}
\end{figure}

This paper adopts probability graph as the underlying knowledge base between the BS and users. A probability graph summarizes statistical insights from diverse knowledge graphs \cite{9416312}, expanding the conventional triple knowledge graph, as shown in Fig.~\ref{fg:kg}, by introducing the dimension of relational probability. An illustrative example of a probability graph is depicted in Fig.~\ref{fg:pg}. A typical knowledge graph comprises numerous triples, each taking the form
\begin{equation}
    \varepsilon = (h, r, t),
\end{equation}
where $h$ denotes the head entity, $t$ represents the tail entity, and $r$ stands for the relation connecting $h$ and $t$. Unlike the determined relations in traditional knowledge graphs, relations in a probability graph are assigned with corresponding specific probabilities, which can represent the likelihood of a particular relation occurring under specific conditions with fixed head and tail entities.

In the PSC network, we assume that the BS possesses extensive knowledge graphs derived from textual/picture/video data via named entity recognition (NER) \cite{9039685} and relation extraction (RE) \cite{9446853}. Leveraging the shared probability graph between the BS and the users facilitates further compression of high-information-density knowledge graphs.

The probability graph extends the dimensionality of relations by statistically calculating occurrences of diverse relations linked to the same head and tail entities across various samples of knowledge graphs. Employing statistical insights from the probability graph, we can construct a multidimensional conditional probability matrix. This matrix encapsulates the likelihood of a specific triple being valid under the condition that other specified triples are valid, thereby facilitating the omission of certain relations in the knowledge graph and ensuing data compression. However, it is important to note that achieving a more compact data size requires a lower semantic compression ratio, which in turn requires a higher computation load, leading to a trade-off between communication and computation within the PSC network. The detailed implementation specifics of the probability graph can be found in \cite{10333452}.

Within the PSC network framework, each user shares its individualized local probability graph with the BS. Leveraging these shared probability graphs, the BS undertakes semantic information extraction, compressing the large-sized data $\mathcal{D}_k$ of each user based on their respective probability graphs with unique semantic compression ratio denoted by $\rho_k$, which is defined as the ratio of selected semantic information divided by the original information. Subsequently, the compressed data $\mathcal{C}_k$ is transmitted to user $k$ over a noisy channel. Upon receiving semantic data from the BS, each user conducts semantic inference, leveraging its local probability graph to reconstruct the compressed semantic information. For example, assuming that the probability graph in Fig.~\ref{fg:pg} is shared by the BS and the user, if the BS needs to transmit triple (\textit{Tree, outside, Sidewalk}) to the user, then according to the probability graph, it only needs to transmit (\textit{Tree, $\varnothing$, Sidewalk}) where the relation is omitted. After receiving the compressed information, the user can automatically recover the relation ``\textit{outside}'' based on the shared probability graph, because it is of higher probability.

\subsection{RSMA Model}
In RSMA, the framework involves the segmentation of messages designated for individual users into two distinct components: a common part and a private part. The common elements extracted from all users are aggregated to form a unified message referred to as the common message. This common message, denoted by $s_0$, undergoes encoding using a shared codebook applicable to all users. Consequently, all users perform decoding on the common message $s_0$. Simultaneously, the unique private part of each user, denoted by $s_k$, is subjected to encoding to generate an individualized private stream intended exclusively for that particular user $k$. Consequently, the transmitted signal $\mathbf{x}$ originating from the BS can be expressed as
\begin{equation}
    \mathbf{x}=\sum_{k=1}^K\mathbf{w}_k s_k+\mathbf{w}_0 s_0,
\end{equation}
where $\mathbf{w}_k\in \mathbb{C}^{M\times 1}$ represents the transmit beamforming vector associated with the private message $s_k$ intended for user $k$, and $\mathbf{w}_0\in \mathbb{C}^{M\times 1}$ denotes the transmit beamforming vector corresponding to the common message $s_0$. Note that both the private message $s_k$ and the common message $s_0$ are normalized.

For user $k$, the received signal is expressed as
\begin{equation}
    y_k=\mathbf{h}_k^\mathrm{H}\mathbf{x}+n=\sum_{i=1}^K\mathbf{h}_k^\mathrm{H}\mathbf{w}_i s_i+\mathbf{h}_k^\mathrm{H}\mathbf{w}_0 s_0 + n,
\end{equation}
where $\mathbf{h}_k\in \mathbb{C}^{M\times 1}$ represents the channel between the BS and user $k$, and $n\sim\mathcal{CN}(0,\sigma^2)$ is additive white Gaussian noise (AWGN) with an average noise power of $\sigma^2$.

In RSMA, all users are required to decode the common message. For user $k$, the rate for decoding common message $s_0$ can be given by
\begin{equation}
    r_k^\mathrm{c}=B\log_2\left(1+\frac{\left\vert \mathbf{h}^\mathrm{H}_k \mathbf{w}_0\right\vert^2}{\sum^K_{i=1}\left\vert \mathbf{h}^\mathrm{H}_k \mathbf{w}_i\right\vert^2+ \sigma^2}\right),
\end{equation}
where $B$ is the bandwidth of the BS.

Following the decoding of the common message $s_0$, it is reconstructed and subtracted from the original received signal. Subsequently, each user can decode its private message. For user $k$, the rate for decoding its own private message $s_k$ can be given by
\begin{equation}
    r_k^\mathrm{p}=B\log_2\left(1+\frac{\left\vert \mathbf{h}^\mathrm{H}_k \mathbf{w}_k\right\vert^2}{\sum^K_{i=1,i\ne k}\left\vert \mathbf{h}^\mathrm{H}_k \mathbf{w}_i\right\vert^2+ \sigma^2}\right).
\end{equation}

It is noteworthy that the common message $s_0$ comprises common parts extracted from messages intended for various users. Assuming the rate allocated to user $k$ in the common message is $a_k$, the following constraint must hold:
\begin{equation}\label{cc}
    \sum_{i=1}^K a_i\leq r_k^\mathrm{c}, \forall k \in \mathcal{K}.
\end{equation}
Denote the minimum common message decoding rate among all users by
\begin{equation}
    r_0^\mathrm{c}=\min_{k\in\mathcal{K}}r_k^\mathrm{c}.
\end{equation}
Then, constraint \eqref{cc} can be rewritten as
\begin{equation}
    \sum_{i=1}^K a_i\leq r_0^\mathrm{c}.
\end{equation}
This constraint ensures that all users can successfully decode the common message.

\subsection{Transmission Model}
In the considered PSC network, data undergoes compression with a semantic compression ratio before transmission. The compression ratio for user $k$ is defined as
\begin{equation}
    \rho_k=\frac{\mathrm{size}(\mathcal{C}_k)}{\mathrm{size}(\mathcal{D}_k)},
\end{equation}
where the function $\mathrm{size}(\cdot)$ quantifies the data size in terms of bits.

Consequently, we can obtain a semantic rate for user $k$, denoted by
\begin{equation}
    c_k = \frac{1}{\rho_k}(a_k+r_k^\mathrm{p}),
\end{equation}
which represents the actual transmission rate of the semantic information. The reason of multiplying the conventional transmission rate by $1/\rho_k$ is that one bit of obtained semantic information can convey more than one bit in the original information.

\subsection{Computation Model}
In the considered PSC network, semantic information extraction is executed by the BS based on shared probability graphs to compress the original data $\mathcal{D}_k$ into a smaller-sized data $\mathcal{C}_k$. This process is inherently tied to computational resources, with a crucial observation being that higher compression ratios $\rho_k$ correspond to decreased computation load.

As detailed in \cite{10333452}, the computation load for the probability graph-based PSC network can be described by
\begin{equation}\label{cl}
    f\left(\rho\right)=\left\{\begin{array}{l}
        A_1\rho +B_1, C_1< \rho \leq 1, \\
        A_2\rho +B_2, C_2< \rho \leq C_1, \\
        \vdots \\
        A_D\rho +B_D, C_D\leq \rho \leq C_{D-1}.
    \end{array}\right.,
\end{equation}
where $\{A_1, A_2, \cdots, A_D\} < 0$ denote the slopes, $\{B_1, B_2, \cdots, B_D\} > 0$ represent constant terms ensuring the continue connection of each segment, 
and $\{C_1, C_2, \cdots, C_D\}$ are the boundaries for each segment. These parameters are intricately linked to the characteristics of the involved probability graphs. The segmented structure of the computation load function $f(\rho)$, as described in \eqref{cl}, arises due to the nature of semantic inference involving multiple levels of conditional probabilities. As the semantic compression ratio increases, the higher-level semantic probabilities are engaged, necessitating more computational resources.

\begin{figure}[t]
\centering
\includegraphics[width=\linewidth]{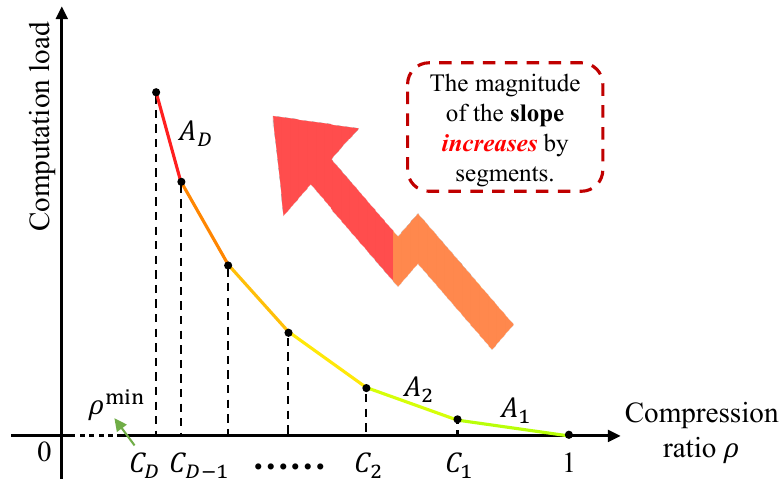}
\caption{Illustration of the computation load versus the compression ratio $\rho$.}
\label{fg:cl}
\end{figure}

The computation load function, denoted by $f(\rho)$, exhibits a segmented architecture with $D$ levels, and the magnitude of the slope escalates in distinct segments, as portrayed in Fig.~\ref{fg:cl}. This behavior stems from the fact that, at higher semantic compression ratios, only low-dimensional conditional probabilities are utilized, resulting in diminished computational requirements. However, as the semantic compression ratio decreases, surpassing particular thresholds like $C_d$, the demand for higher-dimensional information emerges. With increased information dimensions, the computation load intensifies. Each transition within the segmented function $f(\rho)$ signifies the incorporation of probabilistic information with a higher dimension for semantic information extraction.

Given the computation load function $f(\rho)$, the computation power required for extracting semantic information for user $k$ is expressed as
\begin{equation}
    p_k^\mathrm{c} = f(\rho_k)p_0,
\end{equation}
where $p_0$ denotes a positive constant representing the computation power coefficient.

\begin{figure*}[t]
\centering
\includegraphics[width=\linewidth]{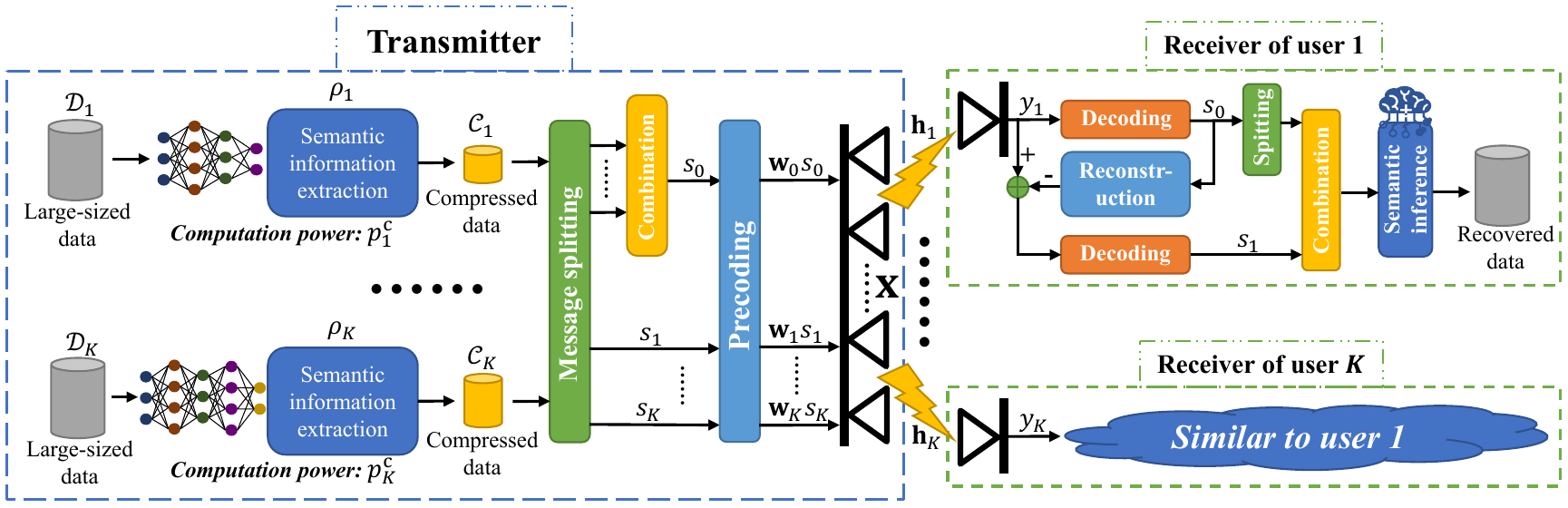}
\caption{The overall framework of the considered PSC network with RSMA.}
\label{fg:fw}
\end{figure*}

The overall framework of the considered PSC network with RSMA is depicted in Fig.~\ref{fg:fw}.

\subsection{Problem Formulation}
Based on the considered system model, our aim is to maximize the sum of semantic rates across all users by jointly optimizing the semantic compression ratio $\rho_k$, rate allocation $a_k$ for each user, and the transmit beamforming $\mathbf{w}_k$ of the BS. This optimization is subject to several constraints, including the maximum total power limit of the BS, rate allocation constraints,  semantic compression constraints, and minimum acceptable semantic rate requirement for each user. The formulated problem for maximizing the sum semantic rate is expressed as
\begin{subequations}\label{pf}
    \begin{align}
        \max_{\mathbf{a},\mathbf{W},\bm{\rho}} \quad & \sum_{k=1}^K c_k, \tag{\ref{pf}}\\
        \textrm{s.t.} \quad & \sum_{k=0}^K \left\vert\mathbf{w}_k\right\vert^2+\sum_{k=1}^K p_k^\mathrm{c}\leq P^\mathrm{max},\label{c1}\\
        & \sum_{k=1}^K a_k\leq r_0^\mathrm{c}, \label{c2}\\
        & a_k\geq 0, \forall k\in\mathcal{K}, \label{c3}\\
        & \rho_{k}^\mathrm{min}\leq\rho_k\leq 1,\forall k\in\mathcal{K},\label{c4}\\
        & c_k\geq c_k^\mathrm{min},\forall k\in\mathcal{K},\label{c5}
    \end{align}
\end{subequations}
where $\mathbf{a}=[a_1,a_2,\cdots,a_K]^\mathrm{T}$, $\mathbf{W}=[\mathbf{w}_0;\mathbf{w}_1;\mathbf{w}_2;\cdots;\mathbf{w}_K]$, $\bm{\rho}=[\rho_1,\rho_2,\cdots,\rho_K]^\mathrm{T}$, $P^\mathrm{max}$ is the maximum power limit of the BS, $\rho_{k}^\mathrm{min}$ is the semantic compression limit for user $k$, and $c_k^\mathrm{min}$ is the minimum acceptable semantic rate of user $k$. Constraint \eqref{c1} imposes a limit on the sum of transmit power and computation power for the BS to stay within the overall power limit $P^\mathrm{max}$. Constraint \eqref{c2} ensures that all users can successfully decode the common message. Constraint \eqref{c3} enforces the non-negativity of the rate allocation. Constraint \eqref{c4} bounds the semantic compression ratio for each user. Lastly, constraint \eqref{c5} ensures that all users attain an acceptable semantic rate.

It is essential to recognize that all these three variables to be optimized are coupled together in problem \eqref{pf}. Furthermore, the objective function of problem \eqref{pf} poses a highly non-convex nature. Another salient feature of problem \eqref{pf} is the inclusion of the segmented function $f(\rho)$, introducing non-smoothness and presenting a distinct challenge to the optimization process. Given the highly non-convex and non-smooth nature of problem \eqref{pf}, obtaining the optimal solution using existing optimization tools in polynomial time is generally challenging. Consequently, we present a suboptimal solution in the following section.

\section{Algorithm Design}
In this section, an iterative algorithm is introduced to address problem \eqref{pf} by optimizing two subproblems: the rate allocation and transmit beamforming design subproblem, and the semantic compression ratio subproblem. These subproblems are alternately optimized within the iterative scheme, progressively converging towards a suboptimal solution for problem \eqref{pf}.

\subsection{Rate Allocation and Transmit Beamforming Design}
With given semantic compression ratio, problem \eqref{pf} can be simplified as
\begin{subequations}\label{pfs1}
    \begin{align}
        \max_{\mathbf{a},\mathbf{W}} \quad & \sum_{k=1}^K \frac{1}{\rho_k}\Bigg[a_k+\notag\\
        & \qquad B\log_2\left(1+\frac{\left\vert \mathbf{h}^\mathrm{H}_k \mathbf{w}_k\right\vert^2}{\sum^K_{i=1,i\ne k}\left\vert \mathbf{h}^\mathrm{H}_k \mathbf{w}_i\right\vert^2+ \sigma^2}\right)\Bigg], \tag{\ref{pfs1}}\\
        \textrm{s.t.} \quad & \sum_{k=0}^K \left\vert\mathbf{w}_k\right\vert^2+\sum_{k=1}^K p_k^\mathrm{c}\leq P^\mathrm{max},\label{s1c1}\\
        & \sum_{k=1}^K a_k\leq r_0^\mathrm{c}, \label{s1c2}\\
        & a_k\geq 0, \forall k\in\mathcal{K}, \label{s1c3}\\
        & \frac{1}{\rho_k}\left[a_k+B\log_2\left(1+\frac{\left\vert \mathbf{h}^\mathrm{H}_k \mathbf{w}_k\right\vert^2}{\sum^K_{i=1,i\ne k}\left\vert \mathbf{h}^\mathrm{H}_k \mathbf{w}_i\right\vert^2+ \sigma^2}\right)\right]\notag\\
        & \hspace{1.6in}\geq c_k^\mathrm{min},\forall k\in\mathcal{K}.\label{s1c4}
    \end{align}
\end{subequations}

Here, given that the semantic compression ratio for each user is fixed, the total computation power for semantic compression, expressed as
\begin{equation}
    \sum_{k=1}^K p_k^\mathrm{c}=p_0\sum_{k=1}^K f(\rho_k),
\end{equation}
becomes a constant. For ease of notation, we denote the total computation power as $P^\mathrm{comp}$, i.e., $\sum_{k=1}^K p_k^\mathrm{c}=P^\mathrm{comp}$.

Problem \eqref{pfs1} is non-convex due to the non-convex nature of the objective function and constraint \eqref{s1c4}. 
To handle the non-convexity of the objective function \eqref{pfs1},
we introduce new variables, namely $\hat{r_k^\mathrm{c}}$ and $\hat{r_k^\mathrm{p}}$. Consequently, problem \eqref{pfs1} can be reformulated as
\begin{subequations}\label{pfs2}
    \begin{align}
        \max_{\mathbf{a},\mathbf{W},\hat{\mathbf{r}^\mathrm{c}},\hat{\mathbf{r}^\mathrm{p}}} \quad & \sum_{k=1}^K \frac{1}{\rho_k}\left(a_k+\hat{r_k^\mathrm{p}}\right), \tag{\ref{pfs2}}\\
        \textrm{s.t.} \quad & \sum_{k=0}^K \left\vert\mathbf{w}_k\right\vert^2\leq P^\mathrm{max}-P^\mathrm{comp},\label{s2c1}\\
        & \sum_{k=1}^K a_k\leq \hat{r_k^\mathrm{c}}, \forall k\in\mathcal{K}, \label{s2c2}\\
        & a_k\geq 0, \forall k\in\mathcal{K}, \label{s2c3}\\
        & \frac{1}{\rho_k}\left(a_k+\hat{r_k^\mathrm{p}}\right) \geq c_k^\mathrm{min},\forall k\in\mathcal{K},\label{s2c4}\\
        & \hat{r_k^\mathrm{c}} \leq B\log_2\left(1+\frac{\left\vert \mathbf{h}^\mathrm{H}_k \mathbf{w}_0\right\vert^2}{\sum^K_{i=1}\left\vert \mathbf{h}^\mathrm{H}_k \mathbf{w}_i\right\vert^2+ \sigma^2}\right),\notag\\ &\hspace{1.9in}\forall k\in\mathcal{K},\label{s2c5}\\
        & \hat{r_k^\mathrm{p}} \leq B\log_2\left(1+\frac{\left\vert \mathbf{h}^\mathrm{H}_k \mathbf{w}_k\right\vert^2}{\sum^K_{i=1,i\ne k}\left\vert \mathbf{h}^\mathrm{H}_k \mathbf{w}_i\right\vert^2+ \sigma^2}\right),\notag\\ &\hspace{1.9in} \forall k\in\mathcal{K},\label{s2c6}
    \end{align}
\end{subequations}
where $\hat{\mathbf{r}^\mathrm{c}}=\left[\hat{r_1^\mathrm{c}},\hat{r_2^\mathrm{c}},\cdots,\hat{r_K^\mathrm{c}}\right]^\mathrm{T}$ and $\hat{\mathbf{r}^\mathrm{p}}=\left[\hat{r_1^\mathrm{p}},\hat{r_2^\mathrm{p}},\cdots,\hat{r_K^\mathrm{p}}\right]^\mathrm{T}$. The objective function of problem \eqref{pfs2} is now convex. Nevertheless, the non-convexity of constraints \eqref{s2c5} and \eqref{s2c6} renders problem \eqref{pfs2} still non-convex. To address the non-convexity of these constraints, we employ a SCA method to approximate constraints \eqref{s2c5} and \eqref{s2c6}.


For constraint \eqref{s2c6}, we can transform it into
\begin{align}
    \hat{r_k^\mathrm{p}} \leq B\log_2\left(1+\gamma_k\right),\label{sca1}\\
    \left\vert \mathbf{h}^\mathrm{H}_k \mathbf{w}_k\right\vert^2 \geq \gamma_k\alpha_k,\label{sca2}\\
    \sum^K_{i=1,i\ne k}\left\vert \mathbf{h}^\mathrm{H}_k \mathbf{w}_i\right\vert^2+ \sigma^2 \leq \alpha_k,\label{sca3}
\end{align}
where $\gamma_k$ and $\alpha_k$ are non-negative slack variables.

Regarding constraint \eqref{sca2}, note that $\mathbf{h}^\mathrm{H}_k \mathbf{w}_k$ can always be a real number by adjusting the phase of the beamforming vector $\mathbf{w}_k$ \cite{9461768}. Therefore, constraint \eqref{sca2} can be rewritten as
\begin{equation}\label{ga}
    \mathcal{R}\left(\mathbf{h}^\mathrm{H}_k \mathbf{w}_k\right) \geq \sqrt{\gamma_k\alpha_k},
\end{equation}
where $\mathcal{R}\left(c\right)$ represents the real part of complex number $c$. For constraint \eqref{ga}, the left hand side becomes convex. For the right hand side, we use the first-order Taylor series to approximate it, which can be given by
\begin{multline}
    \sqrt{\gamma_k\alpha_k} \approx \sqrt{\gamma_k^{(i)}\alpha_k^{(i)}} + \frac{1}{2}\sqrt{\frac{\gamma_k^{(i)}}{\alpha_k^{(i)}}}\left(\alpha_k-\alpha_k^{(i)}\right)\\
    + \frac{1}{2}\sqrt{\frac{\alpha_k^{(i)}}{\gamma_k^{(i)}}}\left(\gamma_k-\gamma_k^{(i)}\right),
\end{multline}
where $i$ denotes the iteration index. Consequently, constraint \eqref{sca2} can be approximated as
\begin{multline}\label{eq1}
    \mathcal{R}\left(\mathbf{h}^\mathrm{H}_k \mathbf{w}_k\right) \geq \sqrt{\gamma_k^{(i)}\alpha_k^{(i)}}\\
    + \frac{1}{2}\sqrt{\frac{\gamma_k^{(i)}}{\alpha_k^{(i)}}}\left(\alpha_k-\alpha_k^{(i)}\right) + \frac{1}{2}\sqrt{\frac{\alpha_k^{(i)}}{\gamma_k^{(i)}}}\left(\gamma_k-\gamma_k^{(i)}\right),
\end{multline}
which is now convex.

For constraint \eqref{s2c5}, we can also transform it into
\begin{align}
    \hat{r_k^\mathrm{c}} \leq B\log_2\left(1+\delta_k\right),\label{251}\\
    \left\vert \mathbf{h}^\mathrm{H}_k \mathbf{w}_0\right\vert^2 \geq \delta_k\beta_k,\label{252}\\
    \sum^K_{i=1}\left\vert \mathbf{h}^\mathrm{H}_k \mathbf{w}_i\right\vert^2+ \sigma^2 \leq \beta_k,\label{253}
\end{align}
where $\delta_k$ and $\beta_k$ are non-negative slack variables.

Concerning constraint \eqref{252}, note that $\mathbf{h}^\mathrm{H}_k \mathbf{w}_0$ cannot always be a real number by changing the phase of the beamforming vector $\mathbf{w}_0$ since there are multiple $\mathbf{h}_k$ and only one $\mathbf{w}_0$. To handle the non-convexity of constraint \eqref{252}, we first rewrite it as
\begin{equation}\label{2522}
    \left\vert \mathbf{h}^\mathrm{H}_k \mathbf{w}_0\right\vert^2 \geq \frac{1}{4}\left[\left(\delta_k+\beta_k\right)^2-\left(\delta_k-\beta_k\right)^2\right],
\end{equation}
which is equivalent to constraint \eqref{252}.

Then, we apply the first-order Taylor series to both sides of constraint \eqref{2522} to approximate it. Consequently, constraint \eqref{2522} can be reformulated as
\begin{multline}\label{eq2}
    2\mathcal{R}\left(\mathbf{h}^\mathrm{H}_k\mathbf{w}_0^{(i)}\mathbf{h}^\mathrm{H}_k\mathbf{w}_0\right)-\left\vert \mathbf{h}^\mathrm{H}_k \mathbf{w}_0^{(i)}\right\vert^2\\
    \geq \frac{1}{4}\bigg[\left(\delta_k+\beta_k\right)^2-\left(\delta_k^{(i)}-\beta_k^{(i)}\right)\left(\delta_k-\beta_k\right)\\
    +\left(\delta_k^{(i)}-\beta_k^{(i)}\right)^2\bigg],
\end{multline}
which is now convex.

With the aforementioned approximations, the non-convex problem \eqref{pfs1} can be approximated by the following convex problem, formulated as
\begin{subequations}\label{pfs3}
    \begin{align}
        \max_{\mathbf{a},\mathbf{W},\bm{\alpha},\bm{\beta},\bm{\gamma},\bm{\delta}} \quad & \sum_{k=1}^K \frac{1}{\rho_k}\left[a_k+B\log_2\left(1+\gamma_k\right)\right], \tag{\ref{pfs3}}\\
        \textrm{s.t.} \quad & \sum_{k=0}^K \left\vert\mathbf{w}_k\right\vert^2\leq P^\mathrm{max}-P^\mathrm{comp},\label{s3c1}\\
        & \sum_{k=1}^K a_k\leq B\log_2\left(1+\delta_k\right), \forall k\in\mathcal{K}, \label{s3c2}\\
        & \frac{1}{\rho_k}\left[a_k+B\log_2\left(1+\gamma_k\right)\right] \geq c_k^\mathrm{min},\forall k\in\mathcal{K},\label{s3c3}\\
        & \eqref{sca3}, \eqref{eq1}, \eqref{253}, \eqref{eq2}, \label{s3c4}\\
        & a_k, \alpha_k, \beta_k, \gamma_k, \delta_k \geq 0, \forall k\in\mathcal{K}, \label{s3c5}
    \end{align}
\end{subequations}
where $\bm{\alpha}=[\alpha_1,\alpha_2,\cdots,\alpha_K]^\mathrm{T}$, $\bm{\beta}=[\beta_1,\beta_2,\cdots,\beta_K]^\mathrm{T}$, $\bm{\gamma}=[\gamma_1,\gamma_2,\cdots,\gamma_K]^\mathrm{T}$, and $\bm{\delta}=[\delta_1,$ $\delta_2,\cdots,\delta_K]^\mathrm{T}$. Problem \eqref{pfs3} is a convex problem that can be efficiently solved using existing optimization tools.

The algorithm for solving the rate allocation and transmit beamforming design subproblem is summarized in Algorithm~\ref{algo1}.

\begin{algorithm}[ht]
\caption{Rate Allocation and Transmit Beamforming Design Optimization with SCA}\label{algo1}
\begin{algorithmic}[1]
    \STATE Initialize $\mathbf{W}^{(0)}$, $\bm{\alpha}^{(0)}$, $\bm{\beta}^{(0)}$, $\bm{\gamma}^{(0)}$ and $\bm{\delta}^{(0)}$. Set iteration index $i=0$.
    \REPEAT
        \STATE Solve convex problem \eqref{pfs3}.
        \STATE Set $i=i+1$.
        \STATE Update $\mathbf{W}^{(i)}$, $\bm{\alpha}^{(i)}$, $\bm{\beta}^{(i)}$, $\bm{\gamma}^{(i)}$ and $\bm{\delta}^{(i)}$ with the optimal solution of problem \eqref{pfs3}.
    \UNTIL{the objective value of problem \eqref{pfs3} converges.}
    \STATE \textbf{Output}: The optimized rate allocation $\mathbf{a}$ and transmit beamforming $\mathbf{W}$.
\end{algorithmic}
\end{algorithm}

\subsection{Semantic Compression Ratio}
With given rate allocation and transmit beamforming design, problem \eqref{pf} can be simplified as
\begin{subequations}\label{pfs4}
    \begin{align}
        \max_{\bm{\rho}} \quad & \sum_{k=1}^K \frac{1}{\rho_k}\left(a_k+r_k^\mathrm{p}\right), \tag{\ref{pfs4}}\\
        \textrm{s.t.} \quad & p_0\sum_{k=1}^K f(\rho_k)\leq P^\mathrm{max}-\sum_{k=0}^K \left\vert\mathbf{w}_k\right\vert^2,\label{s4c1}\\
        & \rho_{k}^\mathrm{min}\leq\rho_k\leq 1,\forall k\in\mathcal{K},\label{s4c2}\\
        & \frac{1}{\rho_k}\left(a_k+r_k^\mathrm{p}\right)\geq c_k^\mathrm{min},\forall k\in\mathcal{K}.\label{s4c3}
    \end{align}
\end{subequations}
Here, $\left(a_k+r_k^\mathrm{p}\right)$ and $\sum_{k=0}^K \left\vert\mathbf{w}_k\right\vert^2$ are constants since the rate allocation and transmit beamforming are given.

The difficulty in solving problem \eqref{pfs4} arises from the non-smoothness of the computation load function $f(\rho)$. To tackle this difficulty, we introduce auxiliary variable $\theta_{kd}\in \{0,1\}$ to represent the linear segment level of user $k$. If $\theta_{kd} =1$, the computation load function $f(\rho_k)$ falls in the $d$-th linear segment for user $k$, i.e., $f(\rho_k)=A_{d}\rho_k+B_{d}$. Otherwise, we have $\theta_{kd} =0$.
Based on the introduce auxiliary variable $\theta_{kd}$,
we reformulate the computation load function as
\begin{equation}
    f(\rho_k)=\sum_{d=1}^{D} \theta_{kd}(A_{d}\rho_k+B_{d}),
\end{equation}
where $D$ is the number of segments of the piecewise function $f(\rho_k)$, and $\theta_{kd}$ identifies the specific segment within which $\rho_k$ falls.
Since each user only has one linear segment level, we have 
\begin{equation}
    \sum_{d=1}^{D}\theta_{kd}=1.
\end{equation}

Therefore, problem \eqref{pfs4} can be rewrite as
\begin{subequations}\label{pfs5}
    \begin{align}
        \max_{\bm{\Theta},\bm{\rho}} \quad & \sum_{k=1}^K \frac{1}{\rho_k}\left(a_k+r_k^\mathrm{p}\right), \tag{\ref{pfs5}}\\
        \textrm{s.t.} \quad & \sum_{k=1}^K \sum_{d=1}^{D} \theta_{kd}(A_{d}\rho_k+B_{d}) \leq \frac{P^\mathrm{max}-\sum_{k=0}^K \left\vert\mathbf{w}_k\right\vert^2}{p_0}, \label{s5c1}\\
        & \rho_{k}^\mathrm{min}\leq\rho_k\leq 1, \forall k\in\mathcal{K},\label{s5c2}\\
        & \rho_k\leq \frac{a_k+r_k^\mathrm{p}}{c_k^\mathrm{min}}, \forall k\in\mathcal{K},\label{s5c3}\\
        & \theta_{kd}\in \{0,1\}, \forall k\in\mathcal{K}, \label{s5c4}\\
        & \sum_{d=1}^{D}\theta_{kd}=1, \forall k\in\mathcal{K}, \label{s5c5}
    \end{align}
\end{subequations}
where $\bm{\Theta}=[\bm{\theta}_1,\bm{\theta}_2,\cdots,\bm{\theta}_K]$ and $\bm{\theta}_k=[\theta_{k1},\theta_{k2},\cdots,\theta_{kD}]^\mathrm{T}$.

In problem \eqref{pfs5}, the matrix $\bm{\Theta}$ is constructed with integer elements from the set $\{0,1\}$, while the vector $\bm{\rho}$ comprises continuous values. Consequently, addressing problem \eqref{pfs5} poses a big challenge due to its mixed-integer programming nature.

It is crucial to emphasize the inherent interdependence between $\bm{\Theta}$ and $\bm{\rho}$. The determination of $\bm{\rho}$ corresponds directly to the specification of $\bm{\Theta}$. Conversely, a fixed $\bm{\Theta}$ does not uniquely determine $\bm{\rho}$; instead, it serves to restrict the feasible range of $\bm{\rho}$ by identifying the specific segment within which $\bm{\rho}$ lies.

As a result, our approach commences by approximating the semantic compression ratio $\bm{\rho}$ through the determination of $\bm{\Theta}$ in the following manner.

For ease of notation, we define
\begin{equation}\label{midrho}
    \rho_{d}=\frac{C_{d-1}+C_{d}}{2},1\leq d\leq D,
\end{equation}
where ${C_{1},C_{2},\cdots,C_{D}}$ denote the boundaries for each segment in function $f(\rho_k)$, and we set $C_{0}=1$.

The value of $\rho_{d}$, being a fixed quantity, serves as a representative of the midpoint within each segment $d$ of the function $f(\rho_k)$. Utilizing this approximation, problem \eqref{pfs5} can be simplified as
\begin{subequations}\label{pfs6}
    \begin{align}
        \max_{\bm{\Theta}} \quad & \sum_{k=1}^K \frac{a_k+r_k^\mathrm{p}}{\sum_{d=1}^{D}\theta_{kd}\rho_{d}}, \tag{\ref{pfs6}}\\
        \textrm{s.t.} \quad & \sum_{k=1}^K \sum_{d=1}^{D} \theta_{kd}(A_{d}\rho_{d}+B_{d}) \leq \frac{P^\mathrm{max}-\sum_{k=0}^K \left\vert\mathbf{w}_k\right\vert^2}{p_0}, \label{s6c1}\\
        & \sum_{d=1}^{D}\theta_{kd}\rho_{d}\leq \frac{a_k+r_k^\mathrm{p}}{c_k^\mathrm{min}}, \forall k\in\mathcal{K},\label{s6c2}\\
        & \theta_{kd}\in \{0,1\}, \forall k\in\mathcal{K}, \label{s6c3}\\
        & \sum_{d=1}^{D}\theta_{kd}=1, \forall k\in\mathcal{K}. \label{s6c4}
    \end{align}
\end{subequations}
Problem \eqref{pfs6} is a zero-one integer programming problem which can be solved using existing optimization tools.

Upon solving problem \eqref{pfs6}, the specific segment in which the semantic compression ratio of each user resides is determined. This designated segment for user $k$ is denoted by $D_k$ and satisfies
\begin{equation}
    f(\rho_k)=A_{D_k}\rho_k+B_{D_k},C_{D_k}\leq\rho_k\leq C_{D_k-1}.
\end{equation}
With the identification of the segment for $\rho_k$, the computation load function $f(\rho_k)$ transforms from a non-smooth piecewise structure into a conventional linear function.

Therefore, the problem we need to solve becomes
\begin{subequations}\label{pfs7}
    \begin{align}
        \max_{\bm{\rho}} \quad & \sum_{k=1}^K \frac{1}{\rho_k}\left(a_k+r_k^\mathrm{p}\right), \tag{\ref{pfs7}}\\
        \textrm{s.t.} \quad & \sum_{k=1}^K A_{D_k}\rho_k+B_{D_k} \leq \frac{P^\mathrm{max}-\sum_{k=0}^K \left\vert\mathbf{w}_k\right\vert^2}{p_0}, \label{s7c1}\\
        & \rho_k\leq \frac{a_k+r_k^\mathrm{p}}{c_k^\mathrm{min}}, \forall k\in\mathcal{K},\label{s7c2}\\
        & C_{D_k}\leq\rho_k\leq C_{D_k-1}, \forall k\in\mathcal{K}. \label{s7c3}
    \end{align}
\end{subequations}
Next, we employ a greedy algorithm to solve problem \eqref{pfs7}.

For convenience, we define
\begin{equation}
    R_k=a_k+r_k^\mathrm{p},
\end{equation}
which is a constant as the rate allocation and transmit beamforming is given.

Subsequently, we arrange the values of $R_k$ in descending order for each user, creating a set denoted as
\begin{equation}
    \mathcal{S}=\{R_1,R_2,\cdots,R_k,\cdots,R_K\}.
\end{equation}
Here, user $1$ holds the highest value, user $k$ ranks $k$-th in magnitude, and user $K$ possesses the lowest value.

Next, we initialize the semantic compression ratio for each user by
\begin{equation}\label{rho0}
    \rho_k=\min\left\{\frac{a_k+r_k^\mathrm{p}}{c_k^\mathrm{min}},C_{D_k-1}\right\}, \forall k\in\mathcal{K},
\end{equation}
which is the lower bound of $\rho_k$ in accordance with constraints \eqref{s7c2} and \eqref{s7c3}.

Then, we can obtain the following theorem.
\begin{theorem}\label{theorem1}
The optimal semantic compression ratio for user $k$ of problem \eqref{pfs7} in the greedy algorithm is
\begin{equation}\label{rhok}
    \rho_k=\left\{\begin{array}{ll}
        C_{D_k}, &\text{if } p_0(A_{D_k}C_{D_k}+B_{D_k})\leq P_k,\\
        \frac{P_k/p_0-B_{D_k}}{A_{D_k}}, &\text{otherwise},
    \end{array}\right.,
\end{equation}
where $P_k$ is the computation power available to user $k$.
\end{theorem}

\begin{proof}
For user $k$, the partial derivative of the objective function of problem \eqref{pfs7}, denoted by $O(\bm{\rho})$, with respect to its semantic compression ratio, $\rho_k$, is
\begin{equation}
    \frac{\partial O(\bm{\rho})}{\partial \rho_k}=\frac{-\left(a_k+r_k^\mathrm{p}\right)}{\rho_k^2},
\end{equation}
which is always negative since the denominator is always positive and the numerator is always negative. Thus, when the semantic compression ratio of the other $k-1$ users are given, the lower the $\rho_k$, the higher the $O(\bm{\rho})$.

Therefore, in our greedy algorithm, we allocate as much computation power as possible to strong users to minimize their semantic compression ratio. This strategy is adopted because assigning a lower semantic compression ratio to a strong user contributes more significantly to $O(\bm{\rho})$ than to a weak user.

For user $k$, the semantic compression ratios of users stronger than it have already been assigned. Consequently, the remaining computation power allocable to user $k$ can be computed as
\begin{equation}\label{pk}
    P_k=P^\mathrm{max}-\sum_{k=0}^K \left\vert\mathbf{w}_k\right\vert^2-p_0\sum_{i=1,i\ne k}^K A_{D_i}\rho_i+B_{D_i}.
\end{equation}

Then, if the computation power available to user $k$, $P_k$, is greater than the highest power demand in the specific segment of $\rho_k$, i.e., $P_k > p_0(A_{D_k}C_{D_k}+B_{D_k})$, then the semantic compression ratio for user $k$ can be given by
\begin{equation}\label{p1}
    \rho_k=C_{D_k}.
\end{equation}

Otherwise, if $P_k$ is less than the highest power demand in the specific segment of $\rho_k$, i.e. $P_k\leq p_0\left[A_{D_k}C_{D_k}+B_{D_k}\right]$, then the semantic compression ratio for user $k$ can be given by
\begin{equation}\label{p2}
    \rho_k=\frac{P_k/p_0-B_{D_k}}{A_{D_k}}.
\end{equation}

Based on \eqref{p1} and \eqref{p2}, the semantic compression ratio of user $k$ can be determined.
$\hfill\square$
\end{proof}

After determining the semantic compression ratio of user $k$, we move on to the next user $k+1$ until we reach the weakest user $K$. Algorithm \ref{algo2} provides a summary of the greedy algorithm.

\begin{algorithm}[ht]
\caption{Greedy Algorithm for Semantic Compression Ratio Optimization}\label{algo2}
\begin{algorithmic}[1]
    \STATE Initialize $\bm{\rho}$ according to \eqref{rho0}.
    \FOR{$k=1:K$}
        \STATE Calculate $P_k$ according to \eqref{pk}.
        \STATE Update $\rho_k$ according to \eqref{rhok}.
    \ENDFOR
    \STATE \textbf{Output}: The optimized semantic compression ratio $\bm{\rho}$.
\end{algorithmic}
\end{algorithm}

\subsection{Algorithm Analysis}
The overall joint transmission and computation resource allocation algorithm for the considered PSC network with RSMA is outlined in Algorithm \ref{algo3}. The intricacy of solving problem \eqref{pf} resides in addressing the two subproblems at each iteration. For the rate allocation and transmit beamforming design subproblem, it is solved by using the SCA method, where the approximated convex problem \eqref{pfs3} is solved in each iteration. The computational complexity of obtaining the solution for problem \eqref{pfs3} is $\mathcal{O}(U_1^2 U_2)$ \cite{lobo1998applications}, where $U_1=5K+(K+1)M$ denotes the total number of variables, and $U_2=11K+1$ is the total number of constraints. Consequently, the overall computational complexity for solving the rate allocation and transmit beamforming design subproblem is $\mathcal{O}(T K^3 M^2)$, where $T$ is the number of iterations for the SCA method. Assuming the algorithm accuracy is denoted by $\epsilon$, the number of iterations for the SCA method is $\mathcal{O}(\sqrt{K}\log_2(1/\epsilon))$. Thus, the total complexity for solving this subproblem is expressed as $\mathcal{O}(K^{3.5} M^2 \log_2(1/\epsilon))$. As for the semantic compression ratio subproblem, its solution is achieved with significantly lower computational complexity compared to the first subproblem, rendering it negligible. Consequently, the overall complexity of Algorithm \ref{algo3} is denoted as $\mathcal{O}(E K^{3.5} M^2 \log_2(1/\epsilon))$, where $E$ represents the number of outer iterations of Algorithm \ref{algo3}. Consequently, the proposed optimization algorithm can be efficiently computed with polynomial complexity.

\begin{algorithm}[ht]
\caption{Joint Transmission and Computation Resource Allocation for PSC Network with RSMA}\label{algo3}
\begin{algorithmic}[1]
    \STATE Initialize $\mathbf{a}^{(0)}$, $\mathbf{W}^{(0)}$, and $\bm{\rho}^{(0)}$. Set iteration index $j=0$.
    \REPEAT
        \STATE With given $\bm{\rho}^{(j)}$, solve the rate allocation and transmit beamforming design subproblem and obtain the solution $\mathbf{a}^{(j+1)}$ and $\mathbf{W}^{(j+1)}$.
        \STATE With given $\mathbf{a}^{(j+1)}$ and $\mathbf{W}^{(j+1)}$, solve the semantic compression ratio subproblem and obtain the solution $\bm{\rho}^{(j+1)}$.
        \STATE Set $j=j+1$.
    \UNTIL{the objective value of problem \eqref{pf} converges.}
    \STATE \textbf{Output}: The optimized $\mathbf{a}$, $\mathbf{W}$ and $\bm{\rho}$.
\end{algorithmic}
\end{algorithm}

Considering the formidable challenge of theoretically establishing the optimality of problem \eqref{pf}, obtaining a globally optimal solution would demand exponential computational complexity, which is impractical. Therefore, we introduce Algorithm \ref{algo3} to furnish a suboptimal solution for problem \eqref{pf} within polynomial computational complexity.

\section{Simulation Results}
In this section, the performance of our proposed PSC network with RSMA is evaluated. In the simulation, there are 4 users, and the BS is equipped with 8 antennas. The bandwidth of the BS is 10 MHz. Moreover, the computation power coefficient is set to 1, and the maximum power limit of the BS is set to 30 dBm. For the considered semantic compression task, we adopt the same parameters as in \cite{zhao2023joint}. The main system parameters are listed in Table \ref{tb1}.

\begin{table}[ht]
\centering
\caption{Main System Parameters}
\begin{tabular}{|c||c||c|}
    \toprule\hline
    \textbf{Parameter} & \textbf{Symbol}  & \textbf{Value} \\
    \hline
    Number of users & $K$ & 4 \\ \hline
    Number of antennas & $M$ & 8 \\ \hline
    Bandwidth of the BS & $B$ & 10 MHz \\ \hline
    Computation power coefficient & $p_0$ & 1 \\ \hline
    Maximum power limit & $P^\mathrm{max}$ & 30 dBm \\
    \hline\bottomrule
\end{tabular}
\label{tb1}
\end{table}

To enable comparison, we consider two benchmarks: the `PSC-SDMA' scheme, which uses space division multiple access (SDMA) in the PSC network, and the `Non-semantic' scheme, which does not employ semantic communication techniques and allocates all power to transmission. The proposed scheme is labelled as `PSC-RSMA'.

\begin{figure}[t]
\centering
\includegraphics[width=\linewidth]{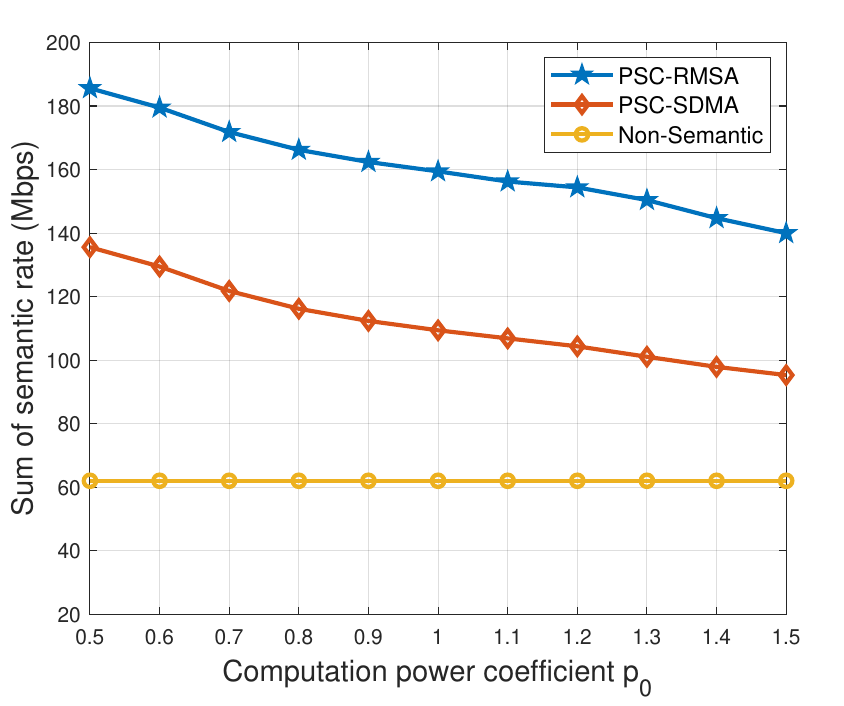}
\caption{Sum of semantic rate vs. computation power coefficient.}
\label{fg:p0}
\end{figure}

Fig.~\ref{fg:p0} displays the relationship between the sum of semantic rate and computation power coefficient. It is evident from the figure that the sum of semantic rate decreases as the computation power coefficient $p_0$ increases for both `PSC' schemes. This is due to the fact that a higher computation power coefficient results in a lower semantic compression ratio for the same computation power consumption, leading to a lower semantic rate. Additionally, the semantic rate of the `PSC-RSMA' scheme is higher than that of the `PSC-SDMA' scheme, indicating the advantage of using common and private rate-splitting in RSMA. Furthermore, the semantic rate of the `Non-semantic' scheme remains constant for different computation power coefficients and is consistently lower than the two `PSC' schemes. This is because it utilizes all the power for transmission and does not perform any semantic computation.

\begin{figure}[t]
\centering
\includegraphics[width=\linewidth]{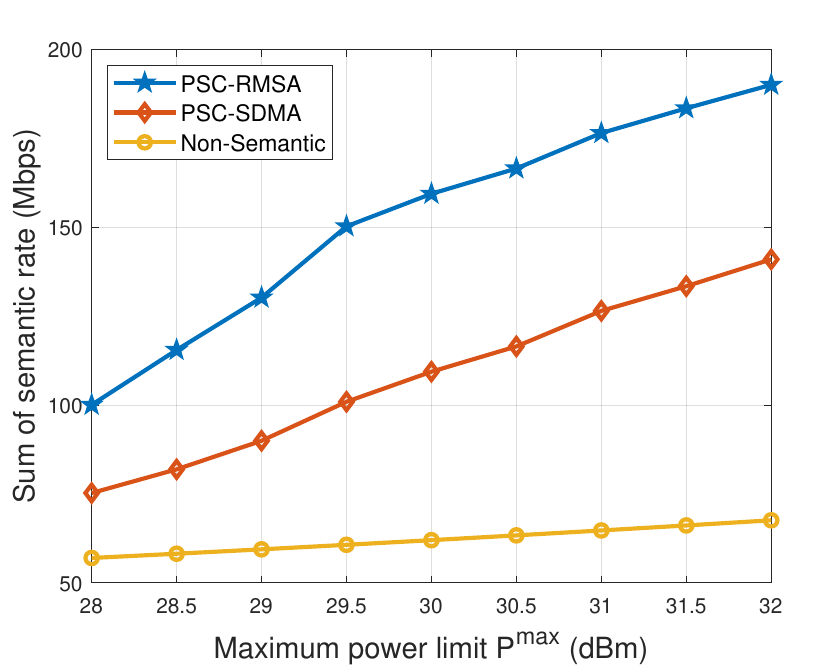}
\caption{Sum of semantic rate vs. maximum power limit.}
\label{fg:pmax}
\end{figure}

Fig.~\ref{fg:pmax} shows the relationship between the maximum power limit of the BS and the sum of semantic rate. As shown in the figure, the sum of semantic rate increases as the maximum power limit $P^\mathrm{max}$ of the BS increases. It is worth noting that the growth rate of the `Non-semantic' scheme is smaller compared to the two `PSC' schemes. The reason for this observation is that when the power budget is sufficient, the `PSC' schemes have the capability to allocate more power to semantic compression, which significantly contributes to the improvement of the semantic rate compared to allocating power solely to transmission. On the other hand, in the `Non-semantic' scheme, allocating all power to transmission not only increases the desired signal power but also introduces more interference between different channels. This interference negatively impacts the overall semantic rate improvement. Hence, the growth rate of the `Non-semantic' scheme is relatively small compared to the `PSC' schemes due to the detrimental effect of increased interference resulting from increased transmission power. In contrast, semantic compression is not affected by interference, leading to a higher growth rate for the `PSC' schemes.

\begin{figure}[t]
\centering
\includegraphics[width=\linewidth]{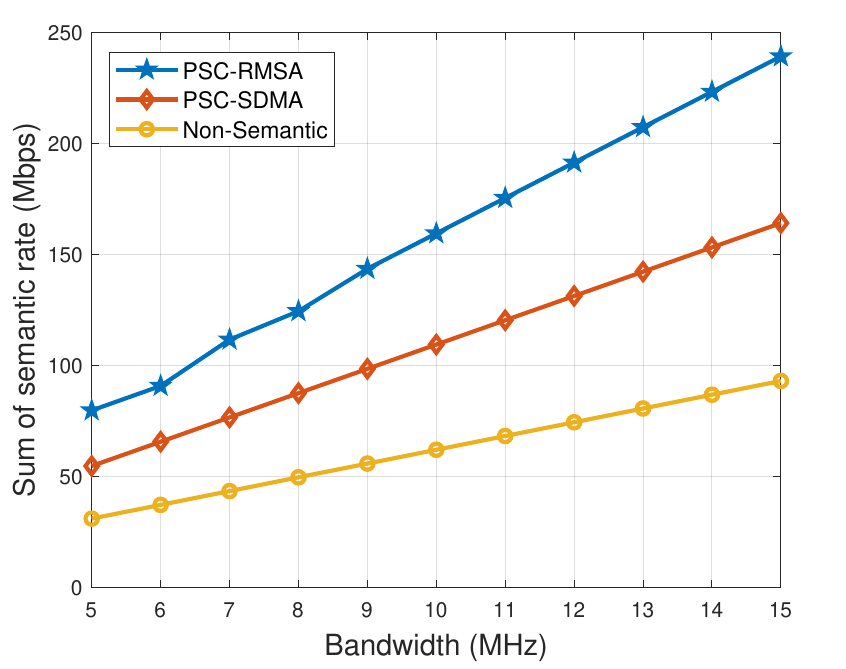}
\caption{Sum of semantic rate vs. bandwidth of the BS.}
\label{fg:B}
\end{figure}

\begin{figure}[t]
\centering
\includegraphics[width=\linewidth]{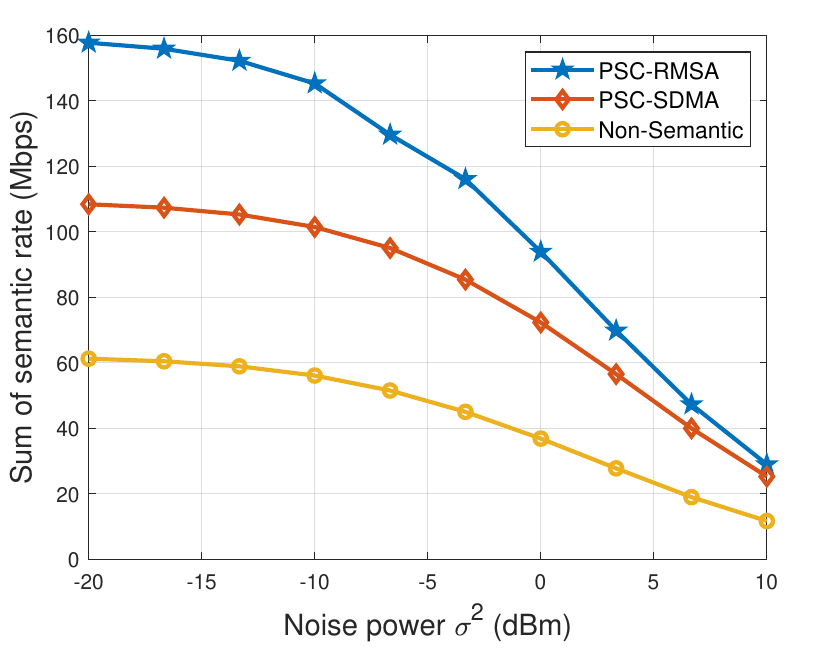}
\caption{Sum of semantic rate vs. noise power.}
\label{fg:np}
\end{figure}

Fig.~\ref{fg:B} illustrates the relationship between the sum of semantic rate and the bandwidth of the BS. It is evident that the sum of semantic rate increases for all schemes as the bandwidth increases. This observation can be attributed to the fact that a higher bandwidth provides more communication resources, resulting in a higher transmission rate. Consequently, the overall semantic rate improves. Fig.~\ref{fg:np} depicts the trend of the sum of semantic rate with respect to the noise power. As depicted in the figure, the sum of semantic rate decreases for all three schemes as the noise power increases, which aligns with intuition. When the noise power is large, the distinction between the three schemes becomes less pronounced. However, the advantages of the `PSC' schemes become more apparent when the noise power is relatively small.

\section{Conclusion}
This paper investigated the problem of joint transmission and computation resource allocation for the PSC network with RSMA. The model considers a BS that needs to transmit substantial knowledge graphs to multiple users, but the transmission is limited by communication resources. To address this, the BS utilizes semantic communication techniques to compress the data. This paper utilized shared probability graphs between the BS and users to enable semantic communication. The BS compresses the original data using the probability graph specific to each user and transmits the compressed data to the respective user. The user then conducts semantic inference to restore the original data. It is crucial to note that the process of semantic compression costs computation power at the BS, whose total power budget is limited. Therefore, it is important to balance the power consumption between transmission and computation. To achieve this, we formulated an optimization problem that aims to maximize the sum of semantic rates of all users under total power, semantic, and RSMA constraints. We proposed an iterative algorithm to obtain a suboptimal solution for this problem. The numerical results validate the effectiveness of our proposed scheme.

\bibliographystyle{IEEEtran}
\bibliography{main}


\end{document}